 \documentclass[a4paper,11pt]{article}

\usepackage[colorlinks=true]{hyperref}

\hypersetup{
    colorlinks,
    citecolor=green,
    filecolor=black,
    linkcolor=red,
    urlcolor=blue
    outlinecolor=black
}

    \usepackage[T1]{fontenc}

\usepackage[colorlinks=true]{hyperref}
\usepackage{amsfonts,amsmath,amsthm}
\usepackage{amssymb}
\usepackage{hyperref}
\usepackage{graphicx}

\newtheorem{thm}{Theorem}
\newtheorem{Def}{Definition}

\newtheorem{rem}{Remark}

\newtheorem{coro}{Corollary}
\newtheorem{lemma}{Lemma}

\newcommand{\EE}[1]{{\text{\normalsize$\mathbb E$}}\left(#1\right)}
\newcommand{\tr}[1]{\text{Tr}\left(#1\right)}

\title{On the generalization of linear least mean squares estimation to quantum systems with non-commutative outputs}
\author{Nina H~Amini
\thanks{Edward L. Ginzton Laboratory, Stanford University, Stanford, CA 94305, USA.}
\thanks{CNRS, Laboratoire des signaux et syst\`{e}mes (L2S), CentraleSup\'{e}lec, 3 rue Joliot-Curie, 91192 Gif-Sur-Yvette, France, {\tt nina.amini@lss.supelec.fr}.}
\and  Zibo Miao
\thanks{ARC Centre for Quantum Computation and Communication Technology, Research School of Engineering, Australian National University, {\tt\small zibo.miao@anu.edu.au}.}
\and Yu Pan
\thanks{ ARC Centre for Quantum Computation and Communication Technology, Research School of Engineering, Australian National University, {\tt\small yu.pan@anu.edu.au}.}
\and Matthew R James
\thanks{ARC Centre for Quantum Computation and Communication Technology, Research School of Engineering, Australian National University, {\tt\small matthew.james@anu.edu.au}.}
\and Hideo Mabuchi
\thanks{Edward L. Ginzton Laboratory, Stanford University, Stanford, CA 94305, USA, {\tt\small hmabuchi@stanford.edu}.}}
\date{}

\begin{document}
\maketitle


\begin{abstract}
 The purpose of this paper is to study the problem of generalizing the Belavkin-Kalman filter to the case where the classical measurement signal is replaced by a fully quantum non-commutative output signal. We formulate a least mean squares estimation problem that involves a non-commutative system as the filter processing the non-commutative output signal. We solve this estimation problem within the framework of non-commutative probability. Also, we find the necessary and sufficient conditions which make these non-commutative estimators physically realizable. These conditions are restrictive in practice.
\end{abstract}


\section{Introduction}
Quantum filtering theory as a fundamental theory in quantum optics, which was implicit in the
   work of Davies in the 1960s~\cite{davies1969quantum,davies1976quantum} concerning open quantum systems and generalized measurement theory, and culminating in the general theory developed and initiated  by Belavkin during the 1980s \cite{Belavkinone,belavkin2004towards,belavkin1992quantum}. 
   The quantum filter is a stochastic differential equation for the conditional state, from which the best estimates of the system observables may be obtained. In related work by Carmichael, the quantum filter is referred to as the stochastic master equation~\cite{HC93,WM10}.
   
One application of the quantum filter, or variants of it,  is in measurement feedback control~\cite{VPB83,EB05,J05,WM10,JK10,sayrin2011real,amini2013feedback}. As in classical control theory, optimal measurement feedback control strategies may be expressed as functions of information states, of which the conditional state is a particular case~\cite{J05,MJ12}. However, feedback control of quantum systems need not involve measurements, and indeed the topic of coherent quantum  feedback is evolving~\cite{NJP09,mabuchi2008coherent,wiseman1994all,james2008h,yanagisawa2003transfer,lloyd2000coherent,MJ12}, though a general theory of optimal design of coherent quantum feedback systems is at present unavailable. 
In coherent quantum feedback control, the controller is also a quantum system, and information flowing in the feedback loop is also quantum (e.g. via a quantum field). 
This type of feedback has recently led to new proposals for quantum memories, quantum error
correction, and ultra-low power classical photonic signal processing~\cite{kerckhoff2010designing,mabuchi2011nonlinear,mabuchi2012qubit,mabuchi2011coherent}.

 The purpose of this paper is to contribute to the knowledge of coherent quantum estimation and control by developing further a non-commutative formulation of the quantum filter given by Belavkin in 1980s~\cite{Belavkinone}. While the main results obtained by Belavkin apply only to the commutative measurement  case, the problem formulation he used was more general. Belavkin's theory of quantum filtering concerns the estimation of the variables of quantum systems conditioned on classical (commutative) measurement records. For linear quantum stochastic systems, Belavkin's filter has the same form as the classical Kalman filter. The Belavkin-Kalman filter is a classical system that processes the incoming measurements to produce the desired estimates. The estimates may be used for monitoring and/or feedback control of the quantum system.

 In our study, we formulate and solve a problem of optimal estimation of a linear quantum system variables  given the non-commutative outputs within the framework of non-commutative probability theory. In particular, we derive a system of non-commutative stochastic differential equations (the non-commutative Belavkin-Kalman filter) that minimizes a least squares error criterion. Such non-commutative filtering equations are well defined mathematically, even if they do not correspond to a physical system. However, if we wish to implement the non-commutative Belavkin-Kalman filter within the class of physically realizable linear quantum stochastic systems (such as linear quantum optical systems), then we find that the conditions for physical realizability impose strong restrictions.  In this paper,
 we find physical realizability conditions for general case and also for some particular cases. These strong physical realizability conditions are a key contribution of this paper.

 We remark  that our contribution here is different from the problem studied in~\cite{vladimirov2013coherent}. Since, in~\cite{vladimirov2013coherent}, the authors propose another physically realizable quantum system considered as a filter, connected to the output of the plant whose dynamics can be determined by minimizing the mean square discrepancy between the plant's state and the output of the filter. Also, they suppose an additional vacuum noise other than the plant's noises in the form of filter's dynamics. 
However, in this paper, we focus firstly on finding the form of linear least mean squares estimators for the non-commutative outputs by temporarily excluding the physical realizability constraints. To do this, we proceed as classical Kalman filtering and Belavkin-Kalman filtering by supposing that the mean squares estimator should satisfy a linear dynamics of innovation processes and we don't suppose any additional vacuum noises other than $dw$ which is the input process of the plant. As such, we obtain the form of least mean squares estimators for non commutative outputs. Then, we seek necessary and sufficient conditions which make such linear least mean squares estimators automatically physically realizable. As we can observe in examples, for some particular forms of plants, we are obliged to add additional vacuum noises to the least mean squares estimators to ensure physical realizability. These estimators which track asymptotically the plant's state and are physically realizable are called coherent observers \cite{Miao2012}. Roughly speaking, coherent least mean squares estimators and observers are another physical systems connected to the main system in cascade~\cite{vuglar2014,Miao2012,NJP09}. We remark that coherent linear least squares estimators and observers  could in principle be used   for coherent feedback control, although this matter is outside the scope of the present paper.

   \medskip
 
This paper is organized as follows. In Section~\ref{sec:three}, we  present general quantum linear stochastic dynamics. In Section~\ref{sec:four}, we obtain non-commutative linear least mean squares estimators for the general linear quantum stochastic dynamics, expressed in Theorem~\ref{thm:main}. In Section~\ref{sec:phys}, we study the physical realizability of such linear least mean squares estimators. The main results of this section are expressed in Theorem~\ref{thm:phys} and Corollaries~\ref{coro:phys1}-\ref{coro:case2}. Moreover, some illustrative examples are provided. Finally, the conclusion is given in Section~\ref{sec:six}.
\section{Quantum linear stochastic dynamics}~\label{sec:three}
Consider linear quantum possibly non-commutative stochastic systems of the form~\cite{james2008h}
\begin{eqnarray}~\label{eq:sdq}
&&dx\left(t\right)=Ax\left(t\right)dt+Bdw\left(t\right),\nonumber\\
&&dy\left(t\right)=Cx\left(t\right)dt+Ddw\left(t\right).
\end{eqnarray}
Here, $A,$ $B,$ $C$ and $D$ are real matrices in $\mathbb R^{n\times n},$ $\mathbb R^{n\times n_w},$ $\mathbb R^{n_y\times n},$ and $\mathbb R^{n_y\times n_w},$ where $n,n_w,n_y$ are positive integers with $n_w\geq n_y$. Also, 

\noindent $x(t)=[x_1(t),\ldots,x_n(t)]^T$ is a vector of self-adjoint possibly non-commutative system variables defined on a Hilbert space $\mathcal H$. The initial state $x(0)$ is Gaussian with state $\rho_0$ and satisfies  the commutation relations~\footnote{The notation $[A,B]$ corresponds to $AB-BA.$}
\begin{equation*}
[x_j(0),x_k(0)]=2i\Theta_{jk},\quad j,k=1,2,\ldots,n
\end{equation*}
where $\Theta$ is a real antisymmetric matrix with components $\Theta_{jk}$ and $i=\sqrt{-1}.$ We assume that the matrix $\Theta$ can take one of the two following forms:
\begin{itemize}
\item Canonical if $\Theta=\textrm{diag}_{\frac{n}{2}}(J),$ with $n$ even or
\item Degenerate canonical if  $\Theta=\textrm{diag}\left(0_{n'\times n'},\textrm{diag}_{\frac{n-n'}{2}}(J)\right),$ with 

$0< n'\leq n$ and $n-n'$ even.~\footnote{Here $0_{m\times n}$ corresponds to $m\times n$ zero matrix.}
\end{itemize}
Here $J$ corresponds to the real skew-symmetric $2\times 2$ matrix 
$$J=\left[\begin{array}{cc}0 & 1\\-1 & 0\end{array}\right],$$
and the "diag" notation corresponds to a block diagonal matrix. Also, $\textrm{diag}_m(J)$ denotes a $m\times m$ block diagonal
matrix with $m$ matrices $J$ on the diagonal.

The noise $dw(t)$ is a vector of self-adjoint quantum noises with It\={o} table
\begin{equation}~\label{eq:fwone}
dw\left(t\right)dw\left(t\right)^{T}=F_wdt,
\end{equation}
where $F_w$ is a non-negative Hermitian matrix~\footnote{The notation $X^T$ corresponds to the transpose of the matrix $X.$} (see e.g.,~\cite{parthasarathy1992introduction,belavkin1991continuous}).  
Indeed, the special case $F_w=I_{n_w\times n_w}$ describes a  classical noise vector $dw$. However, the more general case 
 $$
 F_w=I_{n_w\times n_w}+i\textrm{diag}\left(0_{n'\times n'},\textrm{diag}_{\frac{n_w-n'}{2}}(J)\right)
 $$
 presents $n'$ classical noises and $n_w-n'$ conjugate quantum noises.~\footnote{Here $I_{n\times n}$ is the $n\times n$ identity matrix.}  
 Here, the self-adjoint entries of the vector $w(t)$ which act on the 
Boson Fock space $\mathcal{F}$
are the quantum noises driving the system and they correspond to boson quadratures (see e.g.,~\cite{HP84,belavkin1992quantum,parthasarathy1992introduction}). This determines the following
commutation relations for the noise components
 \begin{equation}~\label{eq:comw}
 [dw(t),dw^T(t)]=2T_w\,dt,
 \end{equation}
 with $T_w=\frac{1}{2}(F_w-F_w^T).$~\footnote{If X and Y are column vectors of operators, the commutator is defined by $[X,Y^T]=XY^T-(YX^T)^T.$}

Similarly, the process $y$ has the following It\={o} table
\begin{equation*}
dy\left(t\right)dy\left(t\right)^{T}=F_ydt,
\end{equation*}
where $F_y$ is a non-negative Hermitian matrix.  
Indeed, the special case $F_y=I_{n_y\times n_y}$ describes a  classical output vector $dy$. However, the more general case 
 $$
 F_y=I_{n_y\times n_y}+i\textrm{diag}\left(0_{n'\times n'},\textrm{diag}_{\frac{n_y-n'}{2}}(J)\right)
 $$
 presents $n'$ classical outputs and $n_y-n'$ conjugate quantum outputs. The commutation relations for the processes $dy$ is determined by the matrix $T_y$ given by the following 
 \begin{equation*}
 [dy(t),dy^T(t)]=2T_y\,dt,
 \end{equation*}
 with $T_y=\frac{1}{2} (F_y-F_y^T).$ Note that we have the following relations
 \begin{equation*}
 F_y=DF_w D^T,\quad\textrm{and}\quad T_y=DT_w D^T.
 \end{equation*}
As discussed in~\cite{james2008h}, we can always set up the following conventions by appropriately enlarging $dw$, $dy$, $B$ and $C$: i)  $n_y$ is even ii) $F_w$ has the following form 
\begin{equation}~\label{eq:noisec}
F_w=I_{n_w\times n_w}+i\textrm{diag}_{\frac{n_w}{2}}(J),
\end{equation}
hence $n_w$ should be even.
\subsection*{Physical realizability of linear QSDEs}
Not all QSDEs of the form~\eqref{eq:sdq} represent the dynamics of physically meaningful open quantum systems. In the case that $\Theta$ is canonical, the system is physically realizable if it presents an open quantum harmonic oscillator. Now we give the formal definition of physical realizability (see e.g.,~\cite[Definition 3.3]{james2008h}).
\begin{Def}
\rm The system~\eqref{eq:sdq} is said to be physically realizable if one of the following holds:
\begin{itemize}
\item $\Theta$ is canonical and Equation~\eqref{eq:sdq} represents the dynamics of an open quantum harmonic oscillator.
\item $\Theta$ is degenerate canonical and there exists an augmentation of Equation~\eqref{eq:sdq} (see more details in~\cite{james2008h}) such that the new QSDEs represent the dynamics of an open quantum harmonic oscillator. 
\end{itemize}
The system~(\ref{eq:sdq}) describes an open quantum harmonic oscillator if there exists a quadratic Hamiltonian $H = \frac{1}{2} x(0)^T R x(0)$,
with a real, symmetric, $n \times n$ matrix $R,$ and a  coupling operator $L = \Lambda x(0)$, with a complex-valued $\frac{1}{2} n_w \times n$ coupling matrix $\Lambda$ 
	such that
	\begin{eqnarray}
	x_k(t) &=& U(t)^\dagger (x_k(0) \otimes 1) U(t),\quad k=1,\cdots,n \nonumber \\
	y_l(t) &=& U(t)^\dagger (1 \otimes w_l(t)) U(t), \quad l=1,\cdots,n_y \label{eqn:unitary}
\end{eqnarray}
where $\{U(t),\quad t\geq 0\}$ is an adapted process of unitary operators satisfying the following QSDE~\cite{HP84}
	\begin{equation*}
	dU(t)=\left(-iH\,dt-\frac{1}{2}L^\dag L\,dt+[-L^\dag\,\, L^T]\Gamma dw(t)\right)U(t),\quad U(0)=I.
	\end{equation*}
In this case, the matrices $A$, $B$, $C$ and $D$ are given by
	\begin{eqnarray*}
		A = 2 \Theta \left(R + \mathfrak{Im}\left(
			\Lambda^{\dagger}\Lambda \right) \right),
			\label{eqn:a2} \\
		B = 2i \Theta \left[ \begin{array}{cc}
			-\Lambda^{\dagger} & \Lambda^T \end{array} \right] \Gamma,
			\label{eqn:b2} \\
		C = \mathcal P^T \left[ \begin{array}{cc}
				\Sigma & 0 \\ 0 & \Sigma \end{array} \right]
			\left[ \begin{array}{cc} \Lambda + \Lambda^\# \\
				-i\Lambda + i\Lambda^\# \end{array} \right] ,
			\label{eqn:c2} \\
		D = \left[\begin{array}{cc}
				I_{n_y \times n_y} &
				0_{n_y \times (n_w - n_y)}
			\end{array} \right].
	\end{eqnarray*}
	Here, 
	$\Gamma$ is a 
	$n_w \times n_w$ matrix and 
	\begin{eqnarray*}
	\Gamma &=&  \mathcal P \textrm{diag} (M), \\ 
	M &=&  \frac{1}{2}
		\left[ \begin{array}{cc} 1 & i \\ 1 & -i \end{array} \right], \\
	\Sigma &=&  \left[ \begin{array}{cc}
				I_{\frac{1}{2}n_y \times 
				\frac{1}{2}n_y} &
				0_{\frac{1}{2}n_y \times \frac{1}{2}\left(
			n_w - n_y \right) } \end{array} \right].
	\end{eqnarray*}
\rm $\mathcal P$ is the appropriately dimensioned square permutation 
	matrix such that 
	$$ \mathcal P \left[ \begin{array}{cccc}a_1 & a_2 & \cdots & a_{2m} \end{array} \right] 
		= \left[ \begin{array}{cccccccc} a_1 & a_3 & \cdots & a_{2m-1} &
		a_2 & a_4 & \cdots & a_{2m} \end{array} \right].$$
	Also, note that $\mathfrak{Im}\left(.\right)$ 
	denotes the imaginary part of a matrix, ${X}^\dagger$ denotes the adjoint of an operator $X,$ and
	${X}^\#$ denotes the complex conjugate of a matrix $X.$
	\end{Def}

The following theorem borrowed from~\cite{james2008h} provides necessary and sufficient conditions for physical realizability of Equation~\eqref{eq:sdq} for any $\Theta$ (canonical or degenerate canonical). 
\begin{thm}[\cite{james2008h}]~\label{thm:prp}
\rm The system~\eqref{eq:sdq} is physically realizable if and only if 
\begin{eqnarray}~\label{eq:pr}
&&iA\Theta+i\Theta A^{T}+BT_wB^{T}=0,\nonumber\\
&&BD^T=\Theta C^T\textrm{diag}_{\frac{n_y}{2}}(J).
\end{eqnarray}
Here, $D=[I_{n_y\times n_y} \quad 0_{n_y\times (n_w-n_y)}].$
Moreover, for canonical $\Theta,$ the Hamiltonian and coupling matrices have explicit expressions as follows. The Hamiltonian matrix $R$ is uniquely given by $R=\frac{1}{4}(-\Theta A+A^T\Theta),$ and the coupling matrix $\Lambda$ is given uniquely by 
\begin{equation*}
\Lambda=-\frac{1}{2}i\,[0_{\frac{n_w}{2}\times \frac{n_w}{2}}\quad I_{\frac{n_w}{2}\times \frac{n_w}{2}}](\Gamma^{-1})^TB^T\Theta.
\end{equation*}
In the case that $\Theta$ is degenerate canonical, a physically realizable augmentation of the system can be constructed to determine the associated Hamiltonian and coupling operators using the above explicit formulas.
\end{thm}
In the following lemma, we prove that the non-demolition property holds for non-commutative outputs if system~\eqref{eq:sdq} is physically realizable.
\begin{lemma}
\rm If the system~\eqref{eq:sdq} is physically realizable, then  the non-demolition property holds, i.e., 
$$[x(t),y(s)^T]=0,\quad\textrm{for any}\quad t\geq s.$$ 
\end{lemma}
\begin{proof}
\rm In~\cite[Lemma 4]{wang2014network}, it was shown that the condition $[x(t),y(s)^T]=0,$ for any $t\geq s$ is equivalent  to the following
\begin{equation}~\label{eq:ndem}
\Theta C^T+BT_wD^T=0. 
\end{equation}
Now we show that if the plant is supposed physically realizable, i.e., if condition~\eqref{eq:pr} holds, then, the above equality holds too. Since, by condition~\eqref{eq:pr}, we have 
\begin{equation*}
BD^T=\Theta C^T\textrm{diag}_{\frac{n_y}{2}}(J),
\end{equation*}
which is equivalent to $\Theta C^T=-BD^T\textrm{diag}_{\frac{n_y}{2}}(J),$ as $(\textrm{diag}_{\frac{n_y}{2}}(J))^2=-I_{n_y\times n_y}.$ Now it is easy to verify that condition~\eqref{eq:ndem} is satisfied, because we have 
$$D^T\textrm{diag}_{\frac{n_y}{2}}(J)=T_w D^T,$$ therefore
$$BD^T\textrm{diag}_{\frac{n_y}{2}}(J)=BT_w D^T,$$ which is exactly condition~\eqref{eq:ndem}.
\end{proof}
In the following lemma, we show that the non-commutative outputs don't have self-non-demolition property.
\begin{lemma}
\rm The non-commutative output processes $y$ have the following commutation relations
\begin{equation*}
[y(t),y(s)^T]=2DT_wD^Ts,\quad\textrm{for all}\quad t\geq s.
\end{equation*} 
\end{lemma}
\begin{proof}
First note that we have the following property
\begin{equation*}
[y(t),y(s)^T]=[y(s)+\int_s^t dy(s'),y(s)^T]=[y(s),y(s)^T],
\end{equation*}
since
\begin{equation*}
[\int_s^t dy(s'),y(s)^T]=[\int_s^t C x(s')+Ddw(s'),y(s)^T]=0,
\end{equation*}
as $[x(s'),y(s)^T]=0$ for any $s'\geq s,$ by previous lemma and $[dw(s'),y(s)^T]=0,$ for any $s'\geq s.$
Now it is sufficient to prove the lemma for $[y(s),y(s)^T].$ By taking the differentiation of this commutator, by It\={o} rule, we find 
\begin{align*}
d[y(s),y(s)^T]&=[dy(s),y(s)^T]+[y(s),dy(s)^T]+[dy(s),dy(s)^T]\\
&=[Cx(s)ds+Ddw(s),y(s)^T]+[y(s),x(s)^TC^Tds+dw(s)^TD^T]\\
&+[Cx(s)ds+Ddw(s),x(s)^TC^Tds+dw(s)^TD^T]\\
&=D[dw(s),dw(s)^T]D^T=2DT_w D^Tds,
\end{align*}
where we have used the following facts: $[x(s),y(s)^T]=0,$ $[y(s),x(s)^T]=0,$ $[dw(s),y(s)^T]=0,$ $[y(s),dw(s)^T]=0,$ $(ds)^2=0,$ $dw(s)ds=0,$ and $ds dw(s)^T=0.$ Also, for the last equality, we have used the commutation relations for the processes $dw$ given in~\eqref{eq:comw}.

Finally, we get the following
$$[y(s),y(s)^T]=2DT_wD^Ts,$$
since $[y(0),y(0)^T]=0_{n_y\times n_y}.$
\end{proof}
\begin{rem}
\rm We recall that when $y$ is commutative, we have $$T_y=DT_wD^T=0_{n_y\times n_y},$$ which implies that the process $y$ is self-commuting.
\end{rem}
Before starting the next section, let us present the following definition. 
\begin{Def}
\rm For any vector of zero mean self-adjoint operators $\zeta,$ the symmetric covariance is defined by
\begin{eqnarray}~\label{eq:symcov}
C_\zeta = \frac{1}{2} \mathbb{E}[ \zeta \zeta^T  + (  \zeta \zeta^T   )^T ] .
\end{eqnarray}
The matrix $C_\zeta$ is non-negative,  real and symmetric.  If $\zeta$ does not have zero mean, the covariance is defined by subtracting the mean.
\end{Def}
\section{Linear least mean squares estimation}
\label{sec:four}
In this section, we formulate a linear least squares estimation problem for the non-commutative linear system (\ref{eq:sdq}), with non-commutative output process  $y(t)$.   The problem concerns finding an operator $\hat x(t)$, called an estimator,  such that the dynamical evolution of $\hat x(t)$ depends causally on the output process $y(t)$ and the length of the error
\begin{eqnarray}~\label{eq:error}
e(t) = x(t) - \hat x(t)
\end{eqnarray}
is minimized.  The idea is that $\hat x(t)$ \lq\lq{tracks}\rq\rq  \ the plant operator $x(t)$.
The vector $\hat x(t)$ has self-adjoint operator components defined on a generally  larger space than the system (\ref{eq:sdq}). More precisely, the vector $\hat x$ consists of entries which are self-adjoint operators acting on the tensor product Hilbert Space $\mathcal H\bigotimes \mathcal F\bigotimes \mathcal H_1,$ where $\mathcal H_1$ is the initial Hilbert space of the least squares estimator $\hat x,$ which is a copy of the system space and independent of the system (\ref{eq:sdq}).

Take $\mathcal Y(t)$ for the von Neumann algebra generated by the output process $y(s)$ for $0\leq s \leq t.$  When $\mathcal Y(t)$ is commutative, i.e., $y(t)$ is a classical measurement process (by the Spectral Theorem~\cite[Theorem 3.3]{BHJ07}), the optimal filter in the least
squares sense is obtained by computing the conditional expectation onto $\mathcal Y(t)$~\cite{belavkin1992quantum,BHJ07}. The non-demolition property $([x(t), y(s)^T] = 0,$ for any $t \geq s$) is sufficient to conclude the existence of the commutative conditional expectation~\cite{belavkin1991continuous}.

In contrast to commutative
output, it is not shown whether the least mean squares estimator that we define in Definition~\ref{def:lms}, is equivalent to conditional expectation. This problem
is related to the existence of a non-commutative conditional expectation which is not always guaranteed, and we do not consider this matter in this paper  (see more details in~\cite{takesaki1972conditional}).

Firstly, we define the class $\xi$ of linear estimators of the form,
\begin{equation}~\label{eq:qkalman}
d\check x(t)=A\check{x}(t)dt+K (t)( dy(t) - C\check{x}(t)dt), \ \ \check{x}(0)=\hat x_0,
\end{equation}
where $y$ is the adapted process defined in Equation~\eqref{eqn:unitary} (see~\cite{BHJ07,parthasarathy1992introduction,HP84} for a discussion of adapted quantum processes). 
Equation (\ref{eq:qkalman}) has the standard form of an observer or Kalman filter.  The real gain matrix $K(t)$ is to be determined.

The initial condition  $\hat x(0)$ satisfies the commutation relations
$$
\hat x(0) \hat x(0)^T-(\hat x(0) \hat x(0)^T)^T=2i\Theta .
$$ 
The state of $\hat x(0)$ is taken to be $\hat \rho_0$.    Consequently, the initial state of the composite system  (\ref{eq:sdq}) and  (\ref{eq:qkalman}) is the Gaussian state $\rho= \rho_0 \otimes \hat \rho_0$. 
\begin{Def}~\label{def:lms}
\rm A linear least mean squares estimator $\hat x $ for the non-commutative linear system (\ref{eq:sdq}) has the following properties,
\begin{itemize}
 \item it is defined on the class $\xi,$ i.e., it is a linear system of the form (\ref{eq:qkalman}), and
 \item  the real matrix $K(t)$ is chosen to minimize the symmetrized mean squares error defined as follows
\begin{equation}
J(K(t)):=\mathrm{Tr}[P(t)],
\label{eq:mmse}
\end{equation}
where $P(t)$ is the symmetric error covariance matrix defined by 
\begin{equation*}~\label{eq:ricsym}
P(t):=C_{e(t)}=\frac{1}{2} \mathbb{E}_\rho[ e(t) e(t)^T+(e(t)e(t)^T)^T].
\end{equation*}
\end{itemize}
\end{Def}
Prior to state our main theorem, we need the following equations 
\begin{align}~\label{eq:P-ric}
\dot{P}\left(t\right)&=\left(A-K(t)C\right)P(t)+P(t)\left(A-K(t)C\right)^{T}
+(B-K(t)D)(B-K(t)D)^T
\nonumber\\
P(0)&=C_{e(0)}.
\end{align}
Here, $C_{e(0)}$ is the initial symmetric error covariance. 
\begin{thm}~\label{thm:main}
\rm Suppose that the plant~\eqref{eq:sdq} is physically realizable. Then, the linear system (\ref{eq:qkalman})  is a linear least mean squares estimator for the system (\ref{eq:sdq})  if and only if the gain $K(t)$ is given by
\begin{eqnarray}
K(t)=BD^T+P(t)C^T,
\label{eq:K-optimal}
\end{eqnarray}
where $P(t)$ is the symmetric positive definite solution to the Riccati equation (\ref{eq:P-ric}).
Furthermore, the innovations process
\begin{equation}~\label{eq:inn}
dr(t):=dy(t)-C\hat{x}(t)dt=Ce(t)dt+Ddw(t), \ \ r(0)=0,
\end{equation}
is a quantum Wiener process with symmetrized covariance $t$ and    It\={o} table
\begin{eqnarray}
dr(t) dr^T(t) = D F_w D^T dt .
\end{eqnarray}
\end{thm}
\begin{proof}
The proof is a modification of the well known methods documented in \cite{kwakernaak1972linear}. It follows from \cite[Lemma 3.1]{kwakernaak1972linear}  that $J(K)$ is minimized, if and only if 
$$K(t)=(BD^T+P(t)C^T)(DD^T)^{-1},$$ where $P(t)$ is the solution to (\ref{eq:P-ric}). Below, we show that the symmetrized error covariance matrix satisfies the Riccati equation~\eqref{eq:P-ric}.

The error as defined in Equation~\eqref{eq:error} satisfies  the following dynamics
\begin{equation}~\label{eq:sde}
de(t)=\Big(A-K(t)C\Big)e(t)\,dt+\Big(B-K(t)D\Big)\,dw(t).
\end{equation}
Fix any real matrix $K(t)$ and let $\hat x(t)$ be the solution of Equation (\ref{eq:qkalman}). Let $e(t)=x(t)-\hat x(t)$ be the associated error, and 
consider the real symmetric error covariance 
\begin{eqnarray}
P(t) = C_{e(t)}=\frac{1}{2} \mathbb{E}_\rho[ e(t) e(t)^T+(e(t)e(t)^T)^T] .
\end{eqnarray}
Take the derivation of the above equation, by the quantum It\={o} rule, we have
\begin{align*}
dP(t)&=\frac{1}{2}\Big(\mathbb{E}_\rho\Big[ de(t) e(t)^T+e(t)de(t)^T+de(t)de(t)^T\nonumber\\
&+\left(de(t) e(t)^T+e(t)de(t)^T+de(t)de(t)^T\right)^T\Big] \Big)
\end{align*}
Now, it is sufficient to replace the expression of $de$ given in Equation~\eqref{eq:sde} in above, we get
\begin{align}~\label{eq:ricito}
dP(t)&=\frac{1}{2}\Big((A-KC)\mathbb{E}_\rho[e(t)e(t)^T]\,dt+\mathbb{E}_\rho[e(t)e(t)^T](A-KC)^T\,dt\nonumber\\
&+(B-KD)F_w(B-KD)^T\,dt
+\mathbb{E}_\rho[(e(t)e(t)^T)^T](A-KC)^T\,dt\nonumber\\
&+(A-KC)\mathbb{E}_\rho[(e(t)e(t)^T)^T]\,dt+(B-KD)F_w^T(B-KD)^T\,dt\Big),
\end{align}
where we have used the followings: $\mathbb E_\rho[dw(t)e(t)^T]=0,$ $\mathbb E_\rho[e(t)dw(t)^T]=0,$ $\mathbb E_\rho[(dt)^2]=0,$ $\mathbb E_\rho[dtdw^T]=0,$ $\mathbb E_\rho[dwdt]=0,$ and Equation~\eqref{eq:fwone}.

From Equation~\eqref{eq:ricito}, we can write the following
\begin{align}~\label{eq:P-new}
\dot{P}\left(t\right)&=\left(A-K(t)C\right)P(t)+P(t)\left(A-K(t)C\right)^{T}\nonumber\\
&+(B-K(t)D)    (B-K(t)D)^T\nonumber\\
P(0)&=C_{e(0)}
\end{align}
Since by Equation~\eqref{eq:noisec}, we know $\frac{F_w+F_w^T}{2}=I_{n_w\times n_w}.$

Now the mean squares error can be expressed in terms of $P(t)$ as follows
\begin{eqnarray}
J(K(t)) = \mathrm{Tr}[ P(t) ].
\end{eqnarray}
As Equation~\eqref{eq:P-new} has the same form as the standard Riccati equation considered in~\cite{kwakernaak1972linear}, we can apply~\cite[Lemma 3.1]{kwakernaak1972linear} to find the the minimum of $J(K),$ we get
\begin{equation}~\label{eq:kgain}
K(t)=(BD^T+P(t)C^T)(DD^T)^{-1}.
\end{equation} 
The gain given in above can be further simplified as 
$$K(t)=BD^T+P(t)C^T,$$ since $DD^T=I_{n_y\times n_y}$ by physical realizability of the plant. This finishes the proof of the first part of Theorem~\ref{thm:main}.

\medskip

Next, following \cite[Sec. 4.3.6]{kwakernaak1972linear}, let
\begin{eqnarray}
\Gamma(t) = \left[
\begin{array}{cc}
\Gamma_{11}(t)  & \Gamma_{12}(t)
\\
\Gamma_{12}^T(t) & \Gamma_{22}(t)
\end{array}
\right]
\end{eqnarray}
be the (symmetrized) covariance matrix for the vector
$
 \left[
 \begin{array}{c}
r(t)
\\
e(t)
\end{array}
\right].
$
By the It\={o} rule and taking expectations, we find that
{\small\begin{eqnarray*}
\dot \Gamma_{11}(t)   &= & C\Gamma_{12}^T(t) + \Gamma_{12}(t) C^T + I , \ \ \Gamma_{11}(0)=0_{n_y\times n_y},
\\
 \dot \Gamma_{12}(t)   &= &   C \Gamma_{22}(t) + \Gamma_{12}(t) (A-K(t)C)^T + D(B - K(t)D)^T, \ \ \Gamma_{12}(0)=0_{n_y\times n},
 \\
 \dot \Gamma_{22}(t)   &= &  (A-K(t)C) \Gamma_{22}(t)  + \Gamma_{22}(t) (A-K(t)C)^T
 \nonumber \\
 &&  + (B-K(t)D) (B-K(t)D)^T, \ \ \Gamma_{22}(0)=C_{e(0)} .
\end{eqnarray*}}
Comparing with Equation (\ref{eq:P-ric}), we find that $\Gamma_{22}(t)=P(t)$, and from (\ref{eq:K-optimal}), we have
\begin{eqnarray}
 \dot \Gamma_{12}(t)   =   \Gamma_{12}(t) (A-KC)^T, \ \ \Gamma_{12}(0)=0_{n_y\times n},
\end{eqnarray}
which implies $\Gamma_{12}(t)=0_{n_y\times n}$ for all $t \geq 0$. From this, we see that
\begin{eqnarray}
\dot \Gamma_{11}(t) =I_{n_y\times n_y},\ \ \Gamma_{11}(0)=0_{n_y\times n_y},
\end{eqnarray}
and hence $\Gamma_{11}(t) =tI_{n_y\times n_y}.$ Also, it is obvious that we have the following relations for the innovation processes $dr,$
$$dr(t) dr^T(t) = D F_w D^T dt,$$ since $(dt)^2=0,$ $dw(t)dt=0,$ and $dtdw(t)^T=0.$
\end{proof}
\section{Results on physical realizability}~\label{sec:phys}
In this section, we will study the physical realizability of the least mean squares estimators announced in Theorem~\ref{thm:main}. 
In Theorem \ref{thm:main}, we don't assume the linear least mean squares estimator (\ref{eq:qkalman}) to be physically realizable.

Let us write $B$ as 
\begin{equation}~\label{eq:Bsubm}
B=[B'_{n\times n_y}\quad B''_{n\times n_w-n_y}].
\end{equation}
In the following, we will announce a general theorem which gives necessary and sufficient conditions ensuring  physical realizability of the least mean squares estimators given in Theorem~\ref{thm:main}.
 \begin{thm}~\label{thm:phys}
\rm Assume that the plant~\eqref{eq:sdq} satisfies the physical realizability conditions announced in Theorem~\ref{thm:prp}. Then, the linear least mean squares estimator announced in Theorem~\ref{thm:main} is a physical realizable estimator if  and only if
\begin{align}~\label{eq:prc}
&-B\textrm{diag}_{\frac{n_w}{2}}(J)B^T+3B'\textrm{diag}_{\frac{n_y}{2}}(J)B'^T+2PC^T\textrm{diag}_{\frac{n_y}{2}}(J)B'^T\nonumber\\
&+2B'\textrm{diag}_{\frac{n_y}{2}}(J)CP+PC^T\textrm{diag}_{\frac{n_y}{2}}(J)CP=0,
\end{align}
with $P$ satisfying the following Riccati equation
\begin{align}~\label{eq:ric}
\dot{P}(t)&=(A-B'C)P(t)+P(t)(A-B'C)^T-P(t)(C^TC)P(t)+B''B''^T,\nonumber\\
P(0)&=C_{e(0)}
\end{align}
\end{thm}
\begin{proof}
The estimator of the form~\eqref{eq:qkalman} can be rewritten as the following form
\begin{equation*}
d\hat{x}(t)=(A-KC)\hat{x}(t)dt+K (t) dy(t) , \ \ \hat{x}(0)=\hat x_0.
\end{equation*}
If we impose the physical realizability constraints on the estimator
of the form given in above, we get the following condition 
\begin{equation*}
i(A-KC)\Theta+i\Theta (A-KC)^T+KD T_w D^T K^T=0.
\end{equation*}
Now it is sufficient to replace $K$ by its value determined by Theorem~\ref{thm:main} (Equation~\eqref{eq:K-optimal}). We find 
\begin{align}~\label{eq:prg}
&A\Theta+\Theta A^T-BD^T C\Theta-\Theta C^T DB^T-PC^TC\Theta-\Theta C^T C P\nonumber\\
&+BD^T\textrm{diag}_{\frac{n_y}{2}}(J)DB^T+PC^T\textrm{diag}_{\frac{n_y}{2}}(J)CP\nonumber\\
&+BD^T\textrm{diag}_{\frac{n_y}{2}}(J)CP+PC^T\textrm{diag}_{\frac{n_y}{2}}(J)DB^T=0,
\end{align}
where we have used $D\textrm{diag}_{\frac{n_w}{2}}(J)D^T=\textrm{diag}_{\frac{n_y}{2}}(J).$ Now use the following facts
$C\Theta=-\textrm{diag}_{\frac{n_y}{2}}(J)DB^T$ and $A\Theta+\Theta A^T=-B\textrm{diag}_{\frac{n_w}{2}}(J)B^T,$ which are derived from the physical realizability of the plant, i.e., Equation~\eqref{eq:pr}. Also, note that $BD^T=B'.$ From these equalities, Equation~\eqref{eq:prc} can be derived from Equation~\eqref{eq:prg}.

Moreover, Equation~\eqref{eq:ric} is derived from the Riccati equation~\eqref{eq:P-new} by replacing $K$ by its value given in Equation~\eqref{eq:K-optimal} and using the physical realizability of the plant.
\end{proof}
\subsection{Some special cases} 
In the following, first, we study the physical realizability of the least mean squares estimator announced in Theorem~\ref{thm:main} for the case where 

\noindent $B'\textrm{diag}_{\frac{n_y}{2}}(J) B'^T=0,$ with $B'$ defined in Equation~\eqref{eq:Bsubm}. Second, we study the case where $B'=0.$ 
\subsubsection*{Case $\bf 1:$ $\bf B'\textrm{diag}_{\frac{n_y}{2}}(J) B'^T=0$}
 As it is demonstrated in the following corollary, the physical realizability constraint announced in~\eqref{eq:prc} can be simplified. 
 \begin{coro}~\label{coro:phys1}
\rm If $B'\textrm{diag}_{\frac{n_y}{2}}(J) B'^T=0.$ Then, the estimator of the form~\eqref{eq:qkalman} is a physical realizable least mean squares estimator if and only if the following constraints are satisfied.
\begin{itemize}
\item [(i)] $K=B'+PC^T.$
\item[(ii)] For $\Theta$ canonical,
\begin{align}~\label{eq:prc1m}
B''\textrm{diag}_{\frac{n_w-n_y}{2}}(J)B''^T+2P\Theta B'B'^T+2B'B'^T\Theta P=0,
\end{align}
with $P$ the unique symmetric positive definite solution of the following Riccati equation,
\begin{align}~\label{eq:rc}
\dot{P}(t)&=AP(t)+P(t)A^T+P\Theta B'B'^T\Theta P(t)+B''B''^T,\nonumber\\
P(0)&=C_{e(0)}.
\end{align}
\item [(iii)] For $\Theta$ degenerate canonical, 
\begin{itemize}
\item [(i)] if $\textrm{diag}\left(0_{n'\times n'},\textrm{diag}_{\frac{n-n'}{2}}(I)\right)C^T=C^T$, then
\begin{align*}
B''\textrm{diag}_{\frac{n_w-n_y}{2}}(J)B''^T+2P\Theta B'B'^T+2B'B'^T\Theta P=0,
\end{align*}
with $P$ satisfying dynamics~\eqref{eq:rc}.

\item [(ii)] But if $\textrm{diag}\left(0_{n'\times n'},\textrm{diag}_{\frac{n-n'}{2}}(I)\right)C^T\neq C^T$ holds, then, 
\begin{align}~\label{eq:prc2m}
&-B''\textrm{diag}_{\frac{n_w-n_y}{2}}(J)B''^T+2PC^T\textrm{diag}_{\frac{n_y}{2}}(J)B'^T\nonumber\\
&+2B'\textrm{diag}_{\frac{n_y}{2}}(J)CP+PC^T\textrm{diag}_{\frac{n_y}{2}}(J)CP=0,
\end{align}
with $P$ satisfying Equation~\eqref{eq:ric}.
\end{itemize}
\end{itemize}
\end{coro}
\begin{proof}
By Theorem~\ref{thm:phys}, we know that the least mean squares estimator~\eqref{eq:qkalman} is physically realizable if and only if the condition~\eqref{eq:prc} is satisfied. By the assumption $B'\textrm{diag}_{\frac{n_y}{2}}(J) B'^T=0,$ we get 
\begin{align*}
B\textrm{diag}_{\frac{n_w}{2}}(J)B^T=B''\textrm{diag}_{\frac{n_w-n_y}{2}}(J)B''^T.
\end{align*}
Then, the condition~\eqref{eq:prc} becomes
\begin{align}~\label{eq:prcc}
-B''\textrm{diag}_{\frac{n_w-n_y}{2}}(J)B''^T&+2PC^T\textrm{diag}_{\frac{n_y}{2}}(J)B'^T\nonumber\\
&+2B'\textrm{diag}_{\frac{n_y}{2}}(J)CP+PC^T\textrm{diag}_{\frac{n_y}{2}}(J)CP=0.
\end{align}
We know by the physical realizability of the plant $C\Theta=-\textrm{diag}_{\frac{n_y}{2}}(J)DB^T$ 
(and similarly, $\Theta C^T=-BD^T\textrm{diag}_{\frac{n_y}{2}}(J)$).
Now suppose that $\Theta$ is canonical, then by multiplying the above conditions by $\Theta,$ we find
$C=\textrm{diag}_{\frac{n_y}{2}}(J)DB^T\Theta$ (and similarly, $C^T=\Theta BD^T\textrm{diag}_{\frac{n_y}{2}}(J)$). Finally, by replacing the values of $C$ and $C^T$ found as such, we get the constraint~\eqref{eq:prc1m} given in the second part of the corollary.  


Now we consider the case $\Theta$ degenerate canonical, in this case if we multiply the expression of $C\Theta$ (and similarly, $\Theta C^T$) given in above by $\Theta$, we get 
$C\textrm{diag}\left(0_{n'\times n'},\textrm{diag}_{\frac{n-n'}{2}}(I)\right)=\textrm{diag}_{\frac{n_y}{2}}(J)DB^T\Theta$ 
(and similarly, 

\noindent $\textrm{diag}\left(0_{n'\times n'},\textrm{diag}_{\frac{n-n'}{2}}(I)\right)C^T=\Theta BD^T\textrm{diag}_{\frac{n_y}{2}}(J)$). Now it is clear that if the condition $\textrm{diag}\left(0_{n'\times n'},\textrm{diag}_{\frac{n-n'}{2}}(I)\right)C^T=C^T$ holds, then we get exactly the constraint given in the first part. However, if this condition does not hold, from~\eqref{eq:prc}, we get the constraint~\eqref{eq:prcc} which is exactly the condition~\eqref{eq:prc2m}.
\end{proof}
\begin{rem}
\rm We remark that the condition $B'\textrm{diag}_{\frac{n_y}{2}}(J) B'^T=0$ was considered in order to simplify the physical realizability constraints in Equation~\eqref{eq:prc}. As the corollary in above shows, in most of the times, this case is equivalent to eliminating the quadratic terms in Equation~\eqref{eq:prc}. Also, note that if $n_y=n=2,$ the condition $B'\textrm{diag}_{\frac{n_y}{2}}(J) B'^T=0$ is equivalent to the condition $\det(B')=0$, i.e., the quadratures are linearly dependent.  
\end{rem}
\paragraph{Particular case: $n_y=n_w.$} Consider the case $n_y=n_w.$ Then, a physical realizable plant should satisfy $D=I_{n_y\times n_y}.$ As a result, the plant given in~\eqref{eq:sdq} takes the following form 
\begin{eqnarray}~\label{eq:sdqone}
&&dx\left(t\right)=Ax\left(t\right)dt+B\,dw\left(t\right),\nonumber\\
&&dy\left(t\right)=Cx\left(t\right)dt+dw\left(t\right).
\end{eqnarray}
Now let us state the following corollary as analogue of Corollary~\ref{coro:phys1}.
\begin{coro}~\label{coro:phystwo}
\rm The estimator of the form~\eqref{eq:qkalman} associated to the plant's dynamics~\eqref{eq:sdqone} is a physical realizable least mean squares estimator if and only if the following constraints are satisfied
\begin{itemize}
\item [(i)] $K=B+PC^T.$
\item [(ii)] For $\Theta$ canonical, 
\begin{align}~\label{eq:prcone}
2P\Theta BB^T+2BB^T\Theta P=0,
\end{align}
with $P$ satisfying the following Riccati equation
\begin{align}~\label{eq:rcone}
\dot{P}(t)&=AP(t)+P(t)A^T+P(t)\Theta BB^T\Theta P(t),\nonumber\\
P(0)&=C_{e(0)}.
\end{align}
\item [(iii)] For $\Theta$ degenerate canonical, 
\begin{itemize}
\item[(i)] if $\textrm{diag}\left(0_{n'\times n'},\textrm{diag}_{\frac{n-n'}{2}}(I)\right)C^T=C^T$, then
\begin{align*}
2P\Theta BB^T+2BB^T\Theta P=0,
\end{align*}
with $P$ satisfying dynamics~\eqref{eq:rcone}.

\item [(ii)] But if $\textrm{diag}\left(0_{n'\times n'},\textrm{diag}_{\frac{n-n'}{2}}(I)\right)C^T\neq C^T$ holds, then, 
\begin{align}~\label{eq:cons3}
2PC^T\textrm{diag}_{\frac{n_y}{2}}(J)B^T
&+2B\textrm{diag}_{\frac{n_y}{2}}(J)CP\nonumber\\
&+PC^T\textrm{diag}_{\frac{n_y}{2}}(J)CP=0,
\end{align}
with $P$ satisfying the following dynamics
\begin{align}~\label{eq:rics}
\dot{P}(t)&=(A-BC)P(t)+P(t)(A-BC)^T-P(t)(C^TC)P(t)\nonumber\\
P(0)&=C_{e(0)}
\end{align}
\end{itemize}
\end{itemize}
\end{coro}
\begin{proof}
The proof of this corollary can be done by the same arguments provided for Corollary~\ref{coro:phys1}. Since, if $n_y=n_w,$ the condition $B'\textrm{diag}_{\frac{n_y}{2}}(J) B'^T=0$ is equivalent to $B\textrm{diag}_{\frac{n_w}{2}}(J) B^T=0.$
\end{proof}
\paragraph{Particular case: $n=2,$ $n_y=2,$ $n_w=4,$ and $\Theta=J.$} Consider the simple case $n=2,$ $n_y=2,$ $n_w=4,$ and $\Theta=J.$ Take $A=\begin{pmatrix}a_1&a_2\\
a_3&a_4\end{pmatrix},$ $P=\begin{pmatrix}p_1&p_2\\
p_2&p_4\end{pmatrix},$ $B'=\begin{pmatrix}b_1&b_2\\
b_3&b_4\end{pmatrix},$ and $B''=\begin{pmatrix}d_1&d_2\\d_3&d_4\end{pmatrix}.$ 
In the following, we find the constraints which guarantee the physical realizability of the least mean squares estimator announced in Theorem~\ref{thm:main}.
\begin{coro}~\label{coro:spp}
\rm Take $\Theta=J$. If $B'J B'^T=0,$ then, the least mean squares estimator given in Theorem~\ref{thm:main} is physically realizable if and only if  \begin{equation}~\label{eq:phntwo}
2p_1(-b_4^2-b_3^2)+2p_2(2b_1b_3+2b_2b_4)+2p_4(-b_1^2-b_2^2)-\det(B'')=0,\end{equation} with $P$ satisfying the Riccati equation~\eqref{eq:rc}.
\end{coro}
\begin{proof}
The proof can be directly derived from Equation~\eqref{eq:prc1m}. 
\end{proof}
Now, we can conclude the following corollary.
\begin{coro}~\label{coro:spp1}
\rm Suppose $b_1=b_3,$ $b_2=b_4,$ and  $\det(B'')=0.$ Then, the linear least mean squares estimator announced in Theorem~\ref{thm:main}, is physically realizable if and only if 
$p_1+p_4=2p_2.$
\end{coro}
The following corollary shows the difficulty of finding a physical realizable least mean squares estimator for some particular forms of $P,$ $B'$ and $B''$.
\begin{coro}~\label{coro:spp2}
\rm Suppose $b_1=b_3,$ $b_2=b_4,$ $d_1=d_3,$ and $d_2=d_4.$ Then, it is impossible to realize physically a linear least mean squares estimator of the form given in~\eqref{eq:qkalman} such that $p_1=p_2=p_4$.
\end{coro}
\begin{proof}
By Corollary~\ref{coro:spp1}, we know that if $b_1=b_3,$ $b_2=b_4,$ and $\det(B'')=0,$ then the physical realizability condition~\eqref{eq:phntwo} implies that $p_1+p_4=2p_2.$ Thus, when $p_1=p_2=p_4,$ this condition is satisfied.

However, note that a least mean squares estimator should satisfy Equation~\eqref{eq:rc}. Also, we should take into account the facts that $b_1=b_3,$ $b_2=b_4,$ $d_1=d_3,$ and $d_2=d_4.$ Therefore, the steady state solution $P=\lim_{t\to\infty} C_{e(t)}$ of the Riccati equation~\eqref{eq:rc}, if there exists, should satisfy the following 
\begin{align}~\label{eq:consric}
2(a_1p_1+a_2p_2)-(b_1^2+b_2^2)(p_1-p_2)^2+d_1^2+d_2^2&=0\nonumber\\
a_2p_4+a_3p_1+(p_1-p_2)(p_4-p_2)+d_1^2+d_2^2&=0\nonumber\\
2(a_3p_2-a_1p_4)-(b_1^2+b_2^2)(p_2-p_1)^2+d_1^2+d_2^2&=0,
\end{align}
where we have used $a_4=-a_1,$ since by physical realizability of the plant, $A$ should satisfy the following
$$AJ+JA^T=0.$$
We can observe that if $p_1=p_2=p_4,$ the matrix $A$ should have the following form
\begin{equation}~\label{eq:forma}
A=\begin{pmatrix}
a_1&a_2\\
a_3&-a_1
\end{pmatrix},\quad \textrm{with}\quad a_3-a_2=2a_1.
\end{equation}

 Note that in this case $K=B',$ since $PC^T=0.$ Also, we know that $A-KC=A-B'C$ should be a Hurwitz matrix (see e.g.,~\cite{kwakernaak1972linear,speyer2008stochastic}). However, we have  $A-B'C=A,$ since $B'C=B'JB'^TJ=0.$ Now, it remains to show that $A$ with its particular form given in~\eqref{eq:forma} could not be a Hurwitz matrix. It is sufficient to find the eigenvalues of $A.$ We have 
$$\det(A-\lambda I_{2\times 2})=-a_1^2+\lambda^2-a_3 a_2=0,$$ with $a_3-a_2=2a_1.$ This implies that 
$\lambda^2=(a_1+a_2)^2.$
Now, it is clear that $A$ could not be a Hurwitz matrix, i.e., all of its eigenvalues have negative real parts. 
\end{proof}
This result shows that in order to obtain conditions on $B,$ which make the linear least mean squares estimator given in Theorem~\ref{thm:main} physically realizable (for e.g., see Equation~\eqref{eq:phntwo}), we need to suppose some constraints on $P.$ This demonstrates the difficulty to find an appropriate plant whose least mean squares estimator is physically realizable.
\subsubsection*{Case $\bf 2:$ $\bf B'=0$} 
Let us announce the following corollary for this special case.
\begin{coro}~\label{coro:case2}
\rm If $B'=0.$ Then, the estimator of the form~\eqref{eq:qkalman} is a physical realizable least mean squares estimator if and only if  
\begin{itemize}
\item[(I)] For canonical $\Theta,$ we have

$K=0,$ and 
$B''\textrm{diag}_{\frac{n_w-n_y}{2}}(J)B''^T=0;$

\item [(II)] For degenerate canonical $\Theta=\textrm{diag}\left(0_{n'\times n'},\textrm{diag}_{\frac{n-n'}{2}}(J)\right),$ we have
\begin{itemize}
\item [(i)] If $\textrm{diag}\left(0_{n'\times n'},\textrm{diag}_{\frac{n-n'}{2}}(I)\right)C^T=C^T$, then

$K=0,$ and
$B''\textrm{diag}_{\frac{n_w-n_y}{2}}(J)B''^T=0.$
\item [(ii)] 
If $\textrm{diag}\left(0_{n'\times n'},\textrm{diag}_{\frac{n-n'}{2}}(I)\right)C^T\neq C^T,$ then
$K=PC^T,$ and

$-B''\textrm{diag}_{\frac{n_w-n_y}{2}}(J)B''^T+PC^T\textrm{diag}_{\frac{n_y}{2}}(J)CP=0,$ with $P$ satisfying
\begin{align}~\label{eq:riccase3}
\dot{P}(t)&=AP(t)+P(t)A^T-P(t)(C^TC)P(t)+B''B''^T,\nonumber\\
P(0)&=C_{e(0)}.
\end{align}
\item [(iii)] Let us write $C=[C'_{n_y\times n'} \quad C''_{n_y\times (n-n')}].$ Then,
if $$\textrm{diag}\left(0_{n'\times n'},\textrm{diag}_{\frac{n-n'}{2}}(I)\right)C^T\neq C^T$$ and $C'^T\textrm{diag}_{\frac{n_y}{2}}(J)C'=0,$ then
$K=PC^T,$ and

$B''\textrm{diag}_{\frac{n_w-n_y}{2}}(J)B''^T=0,$ with $P$ satisfying the Riccati Equation~\eqref{eq:riccase3}.
\end{itemize}
\end{itemize}
\end{coro}
\begin{proof}
If $\Theta$ is canonical, then we find $K=0,$ since $B'=0$ implies $C=0$ by physical realizability conditions given in Theorem~\ref{thm:prp} (Equation~\eqref{eq:pr}). Then, we can use the results of Corollary~\ref{coro:phys1}, where by replacing $B'=0$ in Equations~\eqref{eq:prc1m} and~\eqref{eq:rc}, we find the conditions given in part (I).

However, if $\Theta$ is degenerate canonical, $C$ is not necessarily zero if $C\Theta=0.$ In this case, we find $K=PC^T.$ 
If $\textrm{diag}\left(0_{n'\times n'},\textrm{diag}_{\frac{n-n'}{2}}(I)\right)C^T= C^T,$ then $C=0.$ This proves the results given in (i) of part (II).

But if $\textrm{diag}\left(0_{n'\times n'},\textrm{diag}_{\frac{n-n'}{2}}(I)\right)C^T\neq C^T,$ then we have to replace $B'=0$ in Equations~\eqref{eq:prc2m} and~\eqref{eq:ric}, but $C$ is not necessarily zero. 
This proves the conditions $(ii)$ of part (II). The condition $(iii)$ in part (II) can be derived from condition $(ii).$ Also, by noting that $C\Theta=0,$ implies $C''=0.$
\end{proof}
Note that $B'=0$ implies $C\Theta=0.$ Roughly speaking, when $C\Theta=0,$ the non-commutative filter obtained in the above theorem, could also be realized with Homodyne  or Hetrodyne detection. Since, no quantum information is transferred from the plant to the filter in this case. This is like the classical filtering cases of Homodyne or Heterodyne detection, where one always ends up taking a single quadrature of the field.
\subsection{Consistency with standard results}
In this subsection, we recall the standard results, i.e., Belavkin-Kalman and classical Kalman filtering. They can be respectively considered as special cases when the output is commutative but the plant's dynamics is non-commutative and when the output and the plant's dynamics are both commutative.
\paragraph{Non-commutative dynamics, commutative (classical) outputs.} It can easily be shown that the least mean squares estimators found in Theorem~\ref{thm:main} are reduced to Belavkin-Kalman filters~\cite{EB05} under the assumptions that Belavkin used, i.e., the commutativity of the outputs and the non-demolition property.

Take $\mathcal{Y}_t$ to be commutative, that is $y(t)$ is self-adjoint for each $t$ and $[y(t), y(s)^T]=0$ for all $s, t$.
By the Spectral Theorem, \cite[Theorem 3.3]{BHJ07}), $y(t)$ corresponds to a classical stochastic process, the measurement process. Such continuous measurement signal arise in Homodyne detection~\cite{WM10}. 

For the commutative outputs, we have the following correlation for the observation process $dy$
\begin{equation*}
dy(t)dy^T(t)=F_y\,dt,
\end{equation*}
with $F_y=I_{n_y\times n_y}.$ Note that we have the following relation between $F_w$ and $F_y,$
$$DF_w D^T=F_y,$$ with $D=[I_{n_y\times n_y} \quad 0_{n_y\times (n_w-n_y)}].$

\medskip

Therefore, $F_w$ takes the following form
\begin{align}~\label{eq:fw}
F_w=I_{n_w\times n_w}+i\textrm{diag}\left(0_{n_y\times n_y},\textrm{diag}_{\frac{n_w-n_y}{2}}(J)\right).
\end{align}

It is well established that the optimal filter satisfies the following dynamics 
\begin{equation}~\label{eq:fopt}
d\hat{x}(t)=A\hat{x}(t)dt+(BD^T+P(t)C^T)( dy(t) - C\hat{x}(t)dt),
\end{equation}
with $\widehat x(0)=\mathbb E[x(0)]=\bar x(0),$ and $$\Sigma_0=\frac{1}{2}\mathbb E\left[(x(0)-\bar x(0))(x(0)-\bar x(0))^T+\Big((x(0)-\bar x(0))(x(0)-\bar x(0))^T\Big)^T\right].$$

Remark that the variable $\widehat x$ has zero commutation relation for any $t\geq 0$, i.e., 
\begin{equation}~\label{eq:com}
\widehat x(t)\widehat x(t)^T-(\widehat x(t)\widehat x(t)^T)^T=0,
\end{equation}
which means that $\widehat x$ is a classical variable. 

The real symmetric matrix $P(t):=\frac{1}{2}\mathbb{E}\left[e\left(t\right)e\left(t\right)^{T}+(e\left(t\right)e\left(t\right)^{T})^T\right]$ satisfies the following Riccati equation
\begin{align*}~\label{eq:P}
\dot{P}\left(t\right)&=\left(A-K(t)C\right)P(t)+P(t)\left(A-K(t)C\right)^{T}
+(B-K(t)D)(B-K(t)D)^T\nonumber\\
P(0)&=\Sigma_0.
\end{align*}
\paragraph{Classical Kalman filtering.} The classical linear stochastic dynamics is described by classical variables as follows
\begin{eqnarray}~\label{eq:sd}
dx(t)&=A x(t)\,dt+B dw\nonumber\\
dy(t)&=C x(t)\,dt+D\, dw,
\end{eqnarray}
where $A,$ $B,$ $C$ and $D$ are real matrices in $\mathbb R^{n\times n},$ $\mathbb R^{n\times n_w}$ and $\mathbb R^{n_y\times n}$ and $\mathbb R^{n_y\times n_w},$ and $w$ is a vector of classical Wiener processes, with $dw(t)dw(t)^T=I_{n_w\times n_w}\,dt.$  Take $\mathcal Y(t)$ as the algebra generated by the observation processes previous to time $t,$ defined by $\mathcal Y(t):=\textrm{span}\{(y(s))_{0\leq s\leq t}\}.$

\medskip

It is well known that the classical Kalman filter~\cite{speyer2008stochastic,kwakernaak1972linear} satisfies the following dynamics
\begin{equation*}
d\widehat x(t)=A\widehat x(t)\,dt+(BD^T+P(t)C^T)(DD^T)^{-1}\left(dy(t)-C\widehat x(t)\,dt\right),
\end{equation*}
with $\widehat x(0)=\mathbb E[x(0)]=\bar x(0)$ and $\mathbb E[(x(0)-\bar x(0))(x(0)-\bar x(0))^T]=\Sigma_0.$

\medskip

The covariance of the error $P(t)=\EE{e(t)e^T(t)}$  satisfies the following Riccati equation
\begin{align}~\label{eq:ricatti}
\dot{P}\left(t\right)&=(A-KC)\,P(t)+P(t)\,(A-KC)^T
+\big(B(t)-K(t)D\big)\big(B(t)-K(t)D\big)^T,\nonumber\\
P(0)&=\Sigma_0.
\end{align}
Note that  for the classical Kalman filter, we have $$F_w=I_{n_w\times n_w},\quad F_y=I_{n_y\times n_y}\quad\textrm{and} \quad\Theta=0_{n\times n}.$$

\subsection{Construction of coherent observers with least mean squares estimators}

Suppose that the linear least mean squares estimator~\eqref{eq:qkalman} does not satisfy the constraints of physical realizability given in Theorem~\ref{thm:phys}. Then, we can allow additional vacuum noise inputs to the least mean squares estimator~\eqref{eq:qkalman} such that the
resulting system is physically realizable. Suppose that the following estimator is physically realizable
\begin{equation*}
d\tilde x=(A-KC)\tilde x+Kdy+bdv,
\end{equation*}
with $dvdv^T=F_v dt,$ where $
F_v=I_{n_v\times n_v}+i\textrm{diag}_{\frac{n_v}{2}}(J),$ and with $n_v$ positive even integer. Also, we suppose that $dv$ is independent of $dw.$ This estimator is called a coherent observer \cite{Miao2012}, since it tracks in average the plant dynamics~\eqref{eq:sdq} when $A-KC$ is Hurwitz, and is physically realizable. The error covariance matrix is defined by the following
\begin{equation*}
\tilde P(t):=C_{\tilde e}(t),
\end{equation*}
where $C_{\tilde e}$ is defined in Equation~\eqref{eq:symcov} and $\tilde e=x-\tilde x$. It is not difficult to show that the error covariance matrix satisfies the following Riccati equation
\begin{align}~\label{eq:rictilde}
\dot{\tilde P}(t)&=(A-B'C)\tilde P(t)+\tilde P(t)(A-B'C)^T-\tilde P(t)(C^TC)\tilde P(t)+B''B''^T+bb^T,\nonumber\\
\tilde P(0)&=C_{\tilde e(0)}.
\end{align}
The steady state solution of the above Riccati equation, if there exists, is given by $\tilde P=\lim_{t\to\infty} C_{\tilde e(t)}.$ Then, the performance can be defined by the following 
$$\tilde J=\tr{\tilde P}.$$
In the following, we give some examples. However, in this paper, we don't discuss different algorithms that can be considered to make least mean squares estimators physically realizable. We choose a matrix $b$ that can make the least mean squares estimator physically realizable and we compare the performance of the estimator $\tilde x$ with the least mean squares estimator $\hat x$ (see e.g.,~\cite{vuglar2014,Miao2012} for more details on different algorithms to design coherent observers).
 \subsection{Examples}
 In the following,  we give some examples from the literature to illustrate the results of this section. Also, these examples show the difficulty to find an example where construction of a physically realizable least mean squares estimator is feasible.
 \paragraph{Example 1.}
Consider an optical cavity of the form 
\begin{align*}
dx(t)&=-\kappa/2\,x(t)\,dt-\sqrt{\kappa}dw(t)\\
dy(t)&=\sqrt{\kappa}\,x(t)\,dt+dw(t),
\end{align*}
where $dw(t)dw(t)^T=(I_{2\times 2}+iJ)\,dt$ and $\Theta=J.$ Therefore, we have $[dy(t),dy^T(t)]=2J\,dt,$ i.e., the output processes are non-commutative.

Take $P=\begin{pmatrix}p_1&p_2\\
p_2&p_4\end{pmatrix}.$ Note that for $\kappa> 0$ arbitrary,  
we get the following Riccati equation
\begin{align*}
\kappa p_1-\kappa p_1^2-\kappa p_2^2&=0,\\
\kappa p_2-\kappa p_1 p_2-\kappa p_2 p_4&=0\\
-\kappa p_2^2+\kappa p_4-\kappa p_4^2&=0.
\end{align*}
This implies that $P=I_{2\times 2}$ and $K=0_{2\times 2}.$ Therefore, we get the following estimator
\begin{align}~\label{eq:llmc}
d\hat x(t)&=-\kappa/2\,\hat x(t)\,dt.
\end{align}
This estimator seems trivial as there is no $dy$ term in the dynamics of the estimator. So no information from the system is used to compute the estimate. As a consequence, it does not matter if $y$ be a commutative or non-commutative process, since $K=0_{2\times 2},$ and $dy,$ it does not appeared in  the dynamics of the estimator. However, note that the estimator is physically realizable if and only if $\kappa=0,$ since $\hat x$ is a process with the commutation $\Theta=J.$ Also, remark that $\kappa=0$ means that the system would be decoupled from the field. Particularly, the estimator~\eqref{eq:llmc} is not useful in practice, as there is no $dy$ term. Hence, there is no interest to make it physically realizable (when $\kappa\neq 0$) by adding some vacuum noises.  
\paragraph{Example 2.} 
Now consider a dynamic squeezer. This is an optical cavity with a non-linear element inside. After appropriate linearizations, an optical squeezer can be described by the
following QSDE (see e.g.,~\cite{gardiner2004quantum,roy2015coherent}) if we assume that $\chi=\chi_r+i\chi_i,$ and $\chi_r=0$,
\begin{align*}
dx&=\begin{pmatrix}
-\frac{1}{2}(\kappa_1+\kappa_2)&-\chi_i\\
-\chi_i&-\frac{1}{2}(\kappa_1+\kappa_2)
\end{pmatrix}xdt-\sqrt{\kappa_1}dw_1-\sqrt{\kappa_2}dw_2\\
dy&=\sqrt{\kappa_1}xdt+dw_1,
\end{align*}
where $dw_1(t)dw_1(t)^T=(I_{2\times 2}+iJ)\,dt,$ $dw_2(t)dw_2(t)^T=(I_{2\times 2}+iJ)\,dt,$ and $\Theta=J.$ We  have the following commutation relations for the output processes, $[dy(t),dy^T(t)]=2J\,dt.$ 

For any arbitrary parameters $\kappa_1\geq 0,$ $\kappa_2\geq 0,$ and $\chi_i,$ reals, the physical realizability constraint~\eqref{eq:prc} is satisfied if and only if 
\begin{equation}~\label{eq:ex2p}
2 \kappa_1 - \kappa_2 - 2 \kappa_1 p_1 - \kappa_1 p_2^2 - 2 \kappa_1 p_4 + \kappa_1 p_1 p_4=0,
\end{equation}
with $P=\begin{pmatrix}
p_1&p_2\\
p_2&p_4
\end{pmatrix}.$ The matrix $P$ should satisfy the Riccati Equation~\eqref{eq:ric}, which becomes  
\begin{align}~\label{eq:ex2ric}
\kappa_2 + (\kappa_1 - \kappa_2) p_1 - \kappa_1 p_1^2 - \kappa_1 p_2^2 - 2 p_2 \chi_i&=0\nonumber\\
(\kappa_1 - \kappa_2) p_2 - \kappa_1 p_1 p_2 - \kappa_1 p_2 p_4 - p_1 \chi_i - p_4 \chi_i&=0\nonumber\\
\kappa_2 - \kappa_1 p_2^2 + (\kappa_1 -  \kappa_2) p_4 - \kappa_1 p_4^2 - 2 p_2 \chi_i&=0.
\end{align} 

 If we take $B'=0$ (with previous notation), i.e., $\kappa_1=0,$ the physically realizable constraint~\eqref{eq:ex2p} is satisfied if and only if $\kappa_2=0.$ This means that both field channels are decoupled from the system. Moreover, if we take $\kappa_1=\kappa_2=0,$ the Riccati Equation~\eqref{eq:ex2ric} has not a unique solution. Also, for $\kappa_2\geq 0$ and $\kappa_1>0,$ the physical realizability condition given in~\eqref{eq:ex2p} imposes a constraint on the form of $P.$ This shows the restrictiveness of physical realizability constraints. 

Now take $\kappa_1=0.1,$ $\kappa_2=0.2,$ and $\chi_i=0.01.$ In this case, we find $P=\begin{pmatrix}
1.0030  & -0.0667\\
   -0.0667  &  1.0030
\end{pmatrix},$ and $K=\begin{pmatrix}
0.0009 &  -0.0211\\
   -0.0211   & 0.0009
\end{pmatrix}.$
Therefore, we get the following least mean squares estimator,
\begin{equation}~\label{eq:ex2}
d\hat x=\begin{pmatrix}
-0.1503 &  -0.0033\\
   -0.0033 &  -0.1503
\end{pmatrix}\hat x dt+\begin{pmatrix}
0.0009  & -0.0211\\
   -0.0211  &  0.0009
\end{pmatrix}dy
\end{equation}
which is not physically realizable.  We have $J(K)=\mathrm {Tr}(P)=2.0060.$

Obviously, we can make the least mean squares estimator~\eqref{eq:ex2} physically realizable by adding a vacuum noise as follows
\begin{equation*}
d\tilde x=\begin{pmatrix}
-0.1503 &  -0.0033\\
   -0.0033 &  -0.1503
\end{pmatrix}\tilde x dt+\begin{pmatrix}
0.0009  & -0.0211\\
   -0.0211  &  0.0009
\end{pmatrix}dy+\begin{pmatrix}
0.5486    &     0\\
         0  &  0.5486
         \end{pmatrix}dv,
\end{equation*}
with $dv(t) dv(t)^T=(I_{2\times 2}+iJ)dt.$ We find $\tilde P=\begin{pmatrix}
1.7955   &-0.0782\\
   -0.0782  &  1.7955
\end{pmatrix}.$ Therefore, $\tilde J(K)=\mathrm {Tr}(\tilde P)=3.5910.$ Remark that the form of the estimator in above is not unique. We recall that the study of different algorithms to design coherent observers is beyond the scope of this paper.  
 \paragraph{Example 3.}
    Consider a degenerate parametric amplifier (DPA) described as follows (in the quadrature representation)
    \begin{align*}
dx&=\begin{pmatrix}
-\frac{1}{2}\kappa+\epsilon_r&\epsilon_i\\
\epsilon_i&-\frac{1}{2}\kappa-\epsilon_r
\end{pmatrix}xdt-\sqrt{\kappa}dw\\
dy&=\sqrt{\kappa}xdt+dw.
\end{align*}
Here, $\Theta=J,$ $dw(t)dw^T(t)=(I_{2\times 2}+iJ)dt,$ then, $[dy(t),dy^T(t)]=2J\,dt.$ Suppose $\kappa\geq 0,$ $\epsilon_r,$ and $\epsilon_i$ are reals. Also, take $P=\begin{pmatrix}
p_1&p_2\\
p_2&p_4
\end{pmatrix}.$ Then, $P$ satisfies the following Riccati equation obtained from Equation~\eqref{eq:ric},
\begin{align}~\label{eq:ex3r}
(2\epsilon_r + \kappa)p_1 - \kappa p_1^2 + 2\epsilon_i p_2 - \kappa p_2^2&=0,\nonumber\\
\epsilon_i p_1 +\kappa p_2  - 
 \kappa p_1 p_2 + \epsilon_i p_4 -\kappa p_2 p_4&=0\nonumber\\
2\epsilon_i p_2 - \kappa p_2^2 + (\kappa - 2\epsilon_r ) p_4 - \kappa p_4^2&=0.
\end{align}
The physical realizability of the least mean squares estimator is satisfied if $\kappa=0,$ since the constraint~\eqref{eq:prc} becomes 
$$2 \kappa - 2 \kappa p_1 - \kappa p_2^2 - 2 \kappa p_4 + \kappa p_1 p_4=0.$$
Once again, if $\kappa=0,$ the system would be decoupled from the field. Moreover, in this case, the Riccati equation~\eqref{eq:ex3r} has not a unique solution. Moreover, if $\kappa>0,$ the physical realizability condition in above imposes a constraint on the form of $P.$ This example also, illustrates the difficulty to find the least mean squares estimator which is physically realizable.

Now take the following parameters: $\kappa=0.1,$ $\epsilon_r=0.01,$ and $\epsilon_i=0.01.$ In this case, we get $P=\begin{pmatrix}
1.2000  &  0.2000\\
    0.2000 &   0.8000\end{pmatrix}$ and $K=\begin{pmatrix}
 0.0632&    0.0632\\
    0.0632&   -0.0632
\end{pmatrix}.$ The performance of the least mean squares estimator is given by $J(K)=\mathrm{Tr}(P)=2.$
    
The linear least mean squares estimator has the following form
\begin{equation*}
d\hat x=\begin{pmatrix}
-0.0600  & -0.0100\\
   -0.0100  & -0.0400
\end{pmatrix}\hat x dt+\begin{pmatrix}
 0.0632&    0.0632\\
    0.0632&   -0.0632\end{pmatrix}dy,
\end{equation*}
which is not  physically realizable. We can make this estimator physically realizable as follows 
\begin{equation*}
d\tilde x=\begin{pmatrix}
-0.0600  & -0.0100\\
   -0.0100  & -0.0400
\end{pmatrix}\tilde x dt+\begin{pmatrix}
 0.0632&    0.0632\\
    0.0632&   -0.0632\end{pmatrix}dy+\begin{pmatrix}
   0.3286   &      0\\
         0  &  0.3286
\end{pmatrix}dv,
\end{equation*}
with $dv(t)dv(t)^T=(I_{2\times 2}+iJ)dt.$ We find $\tilde P=\begin{pmatrix}
1.8034  &  0.1431\\
    0.1431 &   1.5172\end{pmatrix}.$ Therefore $\tilde J(K)=\mathrm{Tr}(\tilde P)=3.3206.$
\paragraph{Example 4.} Consider the following plant 
\begin{align*}
dx&=\begin{pmatrix}
0&\Delta\\
-\Delta&0
\end{pmatrix}x\,dt+\begin{pmatrix}
0&0&0&0\\
0&-2\sqrt{\kappa_2}&0&-2\sqrt{\kappa_3}
\end{pmatrix}\begin{pmatrix}
dw_1\\
dw_2
\end{pmatrix},\\
dy&=\begin{pmatrix}
2\sqrt{\kappa_2}&0\\
0&0
\end{pmatrix}x\,dt+dw_1,
\end{align*}
with $\Theta=J,$ $F_w=I_{4\times 4}+i\textrm{diag}_{2}(J),$ and $F_y=I_{2\times 2}+iJ.$ Then, $[dy(t),dy^T(t)]=2J\,dt,$ which means that the output processes are non-commutative. This plant may
be thought of as representing the scenario of an atom trapped
between two mirrors of a three mirror cavity in the strong coupling
limit in which the cavity dynamics are adiabatically eliminated (see more details in~\cite{doherty1999feedback,gough2007singular}).

Here, $B'=\begin{pmatrix}0&0\\
0&-2\sqrt{\kappa_2}\end{pmatrix}.$ It is trivial that the condition $B'JB'^T=0$ is satisfied. We can easily show that the physical realizability constraint (condition~\eqref{eq:prc}) is reduced to \begin{equation}~\label{eq:ex4pr}
8\kappa_2 p_1=0,
\end{equation} with $P=\begin{pmatrix}p_1&p_2\\
p_2&p_4\end{pmatrix}$ satisfying the Riccati Equation~\eqref{eq:ric}, which takes the following form
\begin{align}~\label{eq:ex4ric}
\kappa_2 - 4 \kappa_2 p_1^2 + 2 \Delta p_2&=0\nonumber\\
-\Delta p_1 - 4 \kappa_2 p_1 p_2 + \Delta p_4&=0\nonumber\\
\kappa_2 - 2 \Delta p_2 - 4 \kappa_2 p_2^2+4\kappa_3&=0.
\end{align}

 So if $\kappa_2=0,$ the physical realizability condition~\eqref{eq:ex4pr} is satisfied. This means that the system should be decoupled from the field  channel $dw_1$. However, if $\kappa_2=0,$ for $\Delta\neq 0,$ the Riccati equation~\eqref{eq:ex4ric} has no solution if $\kappa_3\neq 0.$ Moreover, the Riccati equation~\eqref{eq:ex4ric} has not a unique solution if $\kappa_2=\kappa_3=0.$ Also, if $p_1=0,$ the condition~\eqref{eq:ex4pr} is satisfied. However, it is not difficult to show that there is no positive definite solution to  the Riccati equation above in this case. This illustrates once again the restrictive nature of the physical realizability conditions.

Now take $\kappa_2=\kappa_3=0.1$ and $\Delta=0.01.$ The linear least mean squares estimator takes the following form 
\begin{equation*}
d\hat x=\begin{pmatrix} -0.0883 &0.0100\\
 -0.4001&0
\end{pmatrix}\hat x dt+\begin{pmatrix}0.1397 &0\\0.6168  &-0.6325\end{pmatrix}dy.
\end{equation*}
This estimator is not physically realizable which is also conform with Corollary~\ref{coro:phys1}, since the condition~\eqref{eq:prc1m} is not satisfied in this example.

We can certainly add vacuum noise term $bdv$ (with $dv(t) dv(t)^T=(I_{2\times 2}+iJ)dt$) which is independent of $dw_1$  to the estimator above to make it physically realizable. Therefore, take the following estimator
\begin{equation*}
d\tilde x=\begin{pmatrix} -0.0883 &0.0100\\
 -0.4001&0
\end{pmatrix}\tilde x dt+\begin{pmatrix}0.1397 &0\\0.6168  &-0.6325\end{pmatrix}dy+\begin{pmatrix}
0.4204    &     0\\
         0  &  0.4204
\end{pmatrix}dv.
\end{equation*}
 Let us write the performance for estimators $\hat x$ and $\tilde x$ respectively as follows,
$P=\begin{pmatrix}
0.2208 &  0.9753\\
   0.9753  &  8.8359
\end{pmatrix},$ therefore, $J(K)=9.0567.$

For $\tilde x,$ we find $\tilde P=\begin{pmatrix}
0.7075&    1.1760\\
    1.1760 &  33.9878
    \end{pmatrix},$ then $\tilde J(K)=\tr{\tilde P}=34.6953.$
     \paragraph{Example 5.} Consider the following example which is borrowed from~\cite{wiseman1994all,gough2009series}
\begin{align*}
dx&=\gamma\begin{pmatrix}
-1-\cos(\theta)&\sin(\theta)\\
-\sin(\theta)&-1-\cos(\theta)
\end{pmatrix}xdt+\sqrt{\gamma}\begin{pmatrix}
-1-\cos(\theta)&-\sin(\theta)\\
\sin(\theta)&-1-\cos(\theta)
\end{pmatrix}dw\\
dy&=\sqrt{\gamma}\begin{pmatrix}
1+\cos(\theta)&-\sin(\theta)\\
\sin(\theta)&1+\cos(\theta)
\end{pmatrix}xdt+dw,
\end{align*}
where $dw(t)dw^T(t)=(I_{2\times 2}+iJ)dt,$ $\Theta=J,$ and $[dy(t),dy^T(t)]=2J\,dt.$
This is a simple example of all-optical feedback scheme where the light from one end of a cavity is taken and reflect
it back into the other. For simplicity, it is assumed a bath in the vacuum state and a cavity with equal transmitivities at both end-mirrors.

Our aim is to see whether appropriate parameters $\theta$ and $\gamma$ exist such that the linear least mean squares estimator becomes automatically physically realizable. The  matrix $P=\begin{pmatrix}
p_1&p_2\\
p_2&p_4
\end{pmatrix}$ should satisfy Riccati Equation~\eqref{eq:ric}. 
After some calculations, we get
\begin{align*}
2 (1 + \cos(\theta)) p_1+ 2 \sin(\theta) p_2 - (2 + 2 \cos(\theta)) (p_1^2 + p_2^2)& = 0, \\
  2 (1 + \cos(\theta)) p_2 + \sin(\theta) p_4 -p_1 \sin(\theta) - (2 + 2 \cos(\theta)) (p_1 p_2 + p_2 p_4) &= 0, \\
  2 (1 + \cos(\theta)) p_4 - 2 \sin(\theta) p_2 - (2 + 2 \cos(\theta)) (p_4^2 + p_2^2) &=0.
  \end{align*}
  Thus, we find $P=I_{2\times 2}$ and $K=0_{2\times 2}.$ Therefore, the linear least mean squares estimator has the following form
\begin{equation*}
 d\hat x=\gamma\begin{pmatrix}
-1-\cos(\theta)&\sin(\theta)\\
-\sin(\theta)&-1-\cos(\theta)
\end{pmatrix}\hat xdt,
 \end{equation*}
which is not an interesting estimator in practice as there is no $dy$ term, similar to Example $1.$

Now we find $\theta$ and $\gamma$ such that the least mean squares estimator proposed by Theorem~\ref{thm:main} be physically realizable. To do so, we should solve the following equation which comes from the physical realizability condition given in Equation~\eqref{eq:prc}
\begin{equation*}
\gamma\begin{pmatrix}
0&-2-2\cos(\theta)\\
2+2\cos(\theta)&0
\end{pmatrix}=0_{2\times 2}.
\end{equation*}
This equation is satisfied if $\gamma=0$ or (and)  $\theta=k\pi$ (with $k$ an odd number).  When $\gamma=0,$ this means that the coupling to the field is zero. Also, $\theta$ corresponds to the phase of the vacuum light which is picked up when reflected by the cavity mirror. So, when $\theta=k\pi,$ with $k$ an odd number, this means that the damping through the mirrors can be completely eliminated (see more details in~\cite{wiseman1994all}). Obviously, for these cases, we have $A=B=C=0,$ (with previous notations)  which is meaningless.

  \medskip

We have observed in the examples above, constructing physically realizable least mean squares estimators was impossible when $B'\neq 0$ or we should consider some constraints on the matrix $P$ which makes the problem hard and sometimes impossible to solve. This shows the restrictiveness of the physical realizability constraints. Also, when $B'=0,$ the physically realizable least mean squares estimators are not well defined. Supported by these examples and some others which are not given in this paper, we conclude that maybe it is impossible to find examples which could result in physically realizable least mean squares estimators without any additional quantum noises when $B'\neq 0$. (Note that the case $B'=0$ is not an interesting case, since it could also be realized with Homodyne or Hetrodyne detection, as mentioned before, below Corollary 6.) However, we couldn't show this in general case, maybe it is wrong. Also, note that finding examples is a hard problem since we should solve the quadratic equations in $P$ (Equation~\eqref{eq:ric}) where we obtain $P$ as a function of free parameters of the matrix $A,$ and $B.$ Then, these free parameters could  be determined by replacing $P$ in the physical realizability constraints (Equation~\eqref{eq:prc}).
\section{Conclusion}~\label{sec:six}
We have obtained non-commutative linear least mean squares estimators for linear QSDEs by extending Belavkin-Kalman filters to the case where the output processes are non-commutative.  We have assumed that these least mean squares estimators are given as a linear combination of innovation processes. Furthermore, we studied the physical realizability of such estimators for the general case and some special cases. 

We have observed that when $B'=0,$ it is more simple to construct a physically realizable least mean squares estimator, specially  for $\Theta$ degenerate canonical and when $C^T\textrm{diag}_{\frac{n_y}{2}}(J)C=0$. Since, in this case, the physical realizability condition does not depend on the form of $P$ (see more details in Corollary~\ref{coro:case2}). However, roughly speaking, for this case, the non-commutative filter could also be realized by Homodyne or Hetrodyne detection as $C\Theta=0.$ In general, finding examples which satisfy physical realizability conditions, it is difficult without any assumptions on $P.$ These assumptions create constraints on their associated Riccati equations (see e.g., Theorem~\ref{thm:phys} and Corollaries~\ref{coro:phys1}-~\ref{coro:case2}). Moreover, based on our observations, we can conclude that maybe, the construction of a physically realizable least mean squares estimator without any additional quantum noises is impossible when $B'\neq 0$. Generally speaking, the results presented here show the restrictive nature of physical realizability conditions.

Indeed, this work does not show that the best estimate based on the knowledge of the non-commutative output processes, and under the constraints of the physical realizability, has the form of the proposed linear estimator~\eqref{eq:qkalman}. Further research is required to solve the optimal filtering problem under the non-convex constraints imposed by physical realizability conditions. Furthermore, the optimal filtering problem when the coherent controllers are added into the plant's dynamics (see e.g.,~\cite{NJP09,hamerly2013coherent}) can be considered as a future research plan.

\vspace{0.4cm}

\noindent\bf{Acknowledgments.} \rm The authors gratefully acknowledge Professor T. Duncan for helpful discussions. This work was supported by the Australian Research Council Centre of Excellence for Quantum Computation and Communication Technology (Project No. CE110001027) and Air Force Office of Scientific Research (AFOSR) under Grant No. FA2386-12-1-4075.
The first author would like to thank Ryan Hamerly and Nikolas Tezak for valuable discussions. Nina H. Amini acknowledges the support of Math+X Postdoctoral Fellowship from the
Simons Foundation.

\bibliographystyle{plain} 
\bibliography{ref}

\end{document}